\documentclass[a4paper,12pt]{preprint}
\usepackage{graphics}
\usepackage{amsmath,amssymb,bm,dcolumn,color,amsthm}
\usepackage{cancel}
\usepackage{times}
\usepackage{mhequ}
\usepackage[scanall]{psfrag}
\usepackage{epsfig}
\usepackage[margin=2.5cm]{geometry}
\usepackage{url}
\usepackage{sub_JP}
\usepackage{slashed}
\usepackage{appendix}
\newtheorem{lemma}{Lemma}
\def\fref#1{Fig.~\ref{#1}}
\def\EE{\begin{pmatrix}E\\1\end{pmatrix}}
\def\EET{(E,-1)}
\def\eref#1{(\ref{#1})}
\def\tabref#1{Table~\ref{#1}}

\def\OO{\mathcal{O}}
\newcommand{\pare}[1] {\left( #1 \right)}

\newcommand{\tpsi}{\psi}
\newcommand{\ham}{\mathit{ H}}
\def\TT{\mathcal{T}}
\def\JJ{\mathcal{J}}
\newcommand{\teb}{\TT^E}
\def\sh{\sigma_{\rm H}}

\begin{document}

\title{The Colored Hofstadter Butterfly for the Honeycomb Lattice}
\author{A.~Agazzi${}^1$, J.-P.~Eckmann${}^{1,2}$, and G.M.~Graf${}^{~3}$}
\institute{
${}^1$D\'epartement de Physique Th\'eorique,
Universit\'e de Gen\`eve, CH-1211 Gen\`eve 4 (Switzerland)\\
${}^2$Section de Math\'ematiques,
Universit\'e de Gen\`eve, CH-1211 Gen\`eve 4 (Switzerland)\\
${}^3$Institut f\"ur Theoretische Physik, ETHZ, CH-8093 Z\"urich (Switzerland)
}

\maketitle

\begin{abstract}
We rely on a recent method for determining edge spectra and we use it to compute the Chern numbers for
Hofstadter models on the honeycomb lattice having
rational magnetic flux per unit cell. Based on the bulk-edge
correspondence, the Chern number $\sh $ is given as the winding number of an
eigenvector of a
$2 \times 2$ transfer matrix, 
as a function of the quasi-momentum $k\in (0,2\pi)$. This method is
computationally efficient (of order $\OO(n^4)$ in the resolution of
the desired image). It also shows that for the honeycomb lattice the solution for $\sh $ 
for flux $p/q$ in the $r$-th gap conforms with the Diophantine equation $r=\sh
\cdot p+ s\cdot q$, which determines $\sh\mod q$. A window such as $\sh\in(-q/2,q/2)$, or
possibly shifted, provides a natural further condition for
$\sh$, which however turns out not to be met. Based on extensive numerical calculations, we conjecture that the
solution conforms with the relaxed condition $\sh\in(-q,q)$.
\end{abstract}

\section{Introduction}

The spectral diagram of the Hofstadter model for electrons in the 2-dimensional square
lattice and in presence of a fractional magnetic flux per unit
cell $\Phi/\Phi_0 = p/q$ has become known as the   Hofstadter
butterfly \cite{hofstadter}. The
problem of associating to each spectral gap the corresponding
Chern number, representing the integer quantum Hall conductance $\sh $, 
has been solved in the case of a rectangular lattice potential with
perturbative methods (Thouless \textit{et al.} \cite{TKNN}), by reduction to a Diophantine equation with a simple window condition, the solution of which is unique. The phase diagram representing the values of
$\sh $ as a function of the Fermi energy $E_{\rm F}$ and of $\Phi/\Phi_0$
according to this method has been computed \cite{Avron1} and is
known as the colored Hofstadter butterfly. Recently the colored
Hofstadter butterfly has also been calculated for the triangular lattice
\cite{avron4}. 

Using methods from \cite{ASV} we present results for the analogous problem on the honeycomb lattice, which is in particular the lattice structure of graphene.
The Hofstadter Hamiltonian has been considered for the honeycomb
lattice potential \cite{kreftseiler}, and the corresponding spectral diagram can
be calculated \cite{rammal}. Furthermore, the analysis for the labeling of the 
different phases (gaps) of the spectrum can be generalized to the honeycomb
lattice, leading to exactly the same Diophantine equation. However,
\emph{the constraints on its solutions are not the same as in the rectangular case}, and a
simple algebraic condition for the determination of the Chern numbers
of this problem is still lacking.

We circumvent this problem by computing the Chern numbers of the diagram using
a different approach:   The bulk-edge correspondence \cite{hatsugai, ASV}. By
computing the winding number of the edge eigenstates along a loop in the
Brillouin zone \cite{ASV} it is possible to assign to each gap its Chern number $\sh $. 
The fact that the solution so obtained satisfies
the Diophantine equation is evidence to the reliability of the   
method.

\section{The  natural window condition and its exceptions}

The honeycomb lattice can be viewed as two interpenetrating
triangular lattices (labeled by letters A and B) as displayed in
\fref{fig:graphene}. As a consequence, the wave function of an
electron in the tight binding approximation on this bipartite lattice can be
written as the spinor
\begin{equation*}
\psi_{m,n} = \begin{pmatrix} \psi_{m, n}^A \\ \psi_{m, n}^B \end{pmatrix}
\in \mathbb{C}^2~, \quad \psi = \pare{\psi_{m,n}}_{m,n} \in \mathcal{H} = \ell^2(\mathbb{Z}^2;\mathbb{C}^2)~,
\end{equation*}
where $(m,n)$ label the sites on the Bravais sublattice $A$ (or $B$).
\begin{figure}[]
  \centering
     \def\svgwidth{0.4\textwidth}
  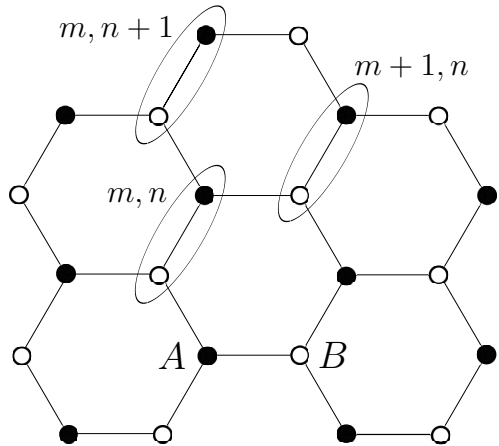
  \caption{Coordinate system on the honeycomb lattice structure.}
  \label{fig:graphene}
\end{figure}
We define the nearest-neighbor (\textit{NN}) \textit{magnetic hopping}
operators $T_i: \mathcal{H} \rightarrow \mathcal{H}, (i = 1,2,3)$ by
\begin{eqnarray} \label{e:gauge}
(T_1 \psi)_{m,n} &=& \begin{pmatrix} \psi_{m, n}^B \\ \psi_{m, n}^A \end{pmatrix}~, \nonumber\\
(T_2 \psi)_{m,n} &=& \begin{pmatrix} \psi^B_{m+1,n} \\ \psi^A_{m-1,n} \end{pmatrix}~,\\
(T_3 \psi)_{m,n} &=& \begin{pmatrix} e^{- 2 \pi i \Phi m} \ \psi^B_{m,n+1} \\ e^{2 \pi i \Phi m} \ \psi^A_{m,n-1} \end{pmatrix}~, \nonumber
\end{eqnarray}
and this fixes our choice of gauge.
The Hofstadter model on the honeycomb lattice
is an isotropic \textit{NN} hopping Hamiltonian
which can be written as
\begin{equation*}
\ham = \sum_{j = 1}^3 T_j~ .
\end{equation*}

Because of the translation invariance of $\ham$ in the $n$-direction, a Bloch decomposition in that direction can be performed:
\begin{equ}
\psi_{m,n} = \int_0^{2 \pi} \kern -1em \textrm{d}k\, \psi_m(k) e^{i k n}~.
\end{equ}
This transformation fibers $\ham$, which can now be written as  
\begin{equ}
\label{e:ham}
(\ham(k) {\psi}(k))_m = A^*(k) {\psi}_{m-1} (k) + V_m(k)
{\psi}_m (k) + A(k) {\psi}_{m+1} (k)~,
\end{equ}
where
\begin{equ}
A(k) = \begin{pmatrix}
   0   &  1  \\
     0 & 0
\end{pmatrix} \quad \textrm{and}\quad
V_{m} (k) =  \begin{pmatrix}
   0   &1 + e^{i  (k - 2 \pi \Phi m)}   \\
      1 + e^{-i  (k - 2 \pi \Phi m)} & 0
\end{pmatrix}~.
\end{equ}
\subsection{The Diophantine equation}
The magnetic field reduces the translation symmetry of the Hamiltonian.
In the case of rational fluxes $\Phi = p/q$, the gauge underlying
\eref{e:gauge} increases the translation invariance period of the Hamiltonian in the
$m$-direction from $1$ to $q$ lattice units. (In the $n$ direction the
translation invariance is not broken.) A  
Bloch decomposition in the $m$-direction can thus be performed and, by plotting the
eigenvalues of the Hamiltonian fiber by fiber, the (black and white) Hofstadter
butterfly of the   
honeycomb lattice is drawn \cite{rammal}.
Moreover, by performing an analysis similar to the one carried out in
\cite{TKNN} or by using that $\ham$ commutes with magnetic translations, one can see that the Hall
conductivity in the $r$-th gap satisfies the Diophantine equation \cite{dana, avron4}
\begin{equation}
\label{diophantine}
r= \sh \cdot p + s \cdot q~.
\end{equation}
This is the same equation as in the square lattice case. 
It clearly
determines $\sh $ up to a $\textrm{mod }q$ ambiguity. 
The
 natural ``window condition''  \cite{TKNN},
\begin{equ}\label{e:window}
\sh  \in \pare{-q/2,q/2}~,
\end{equ}
resolves this ambiguity in the square lattice case.
Unfortunately, \eref{e:window} neither   holds for the
triangular lattice \cite{avron4} nor for the honeycomb lattice, which
we study here. We illustrate this in \tabref{tab:table1} for the
honeycomb lattice. 

\begin{table}[h]
\begin{center} 
\scalebox{1}{
\begin{tabular}{
|c|r@{ }c@{}l|}
\hline
$1/5$ & $1, 2, \leavevmode\cancel{-2}$$\rightarrow$$3, -1, $&$0$&$, 1, \leavevmode\cancel{2}$$\rightarrow$$-3, -2, -1$\\
$1/6$ & $1, 2, 3, \leavevmode\cancel {-2}$$\rightarrow$$4, -1,$& $0$ &$, 1, \leavevmode\cancel {2}$$\rightarrow$$-4, -3, -2, -1$\\
$1/7$ & $1, 2, 3, \leavevmode\cancel {-3}$$\rightarrow$$4, \leavevmode\cancel {-2}$$\rightarrow$$5, -1, $&$0$&$, 1, \leavevmode\cancel{2}$$\rightarrow$$-5, \leavevmode\cancel {3}$$\rightarrow$$-4, -3, -2, -1$\\
$2/5$ & $\leavevmode\cancel{-2}$$\rightarrow$$3, 1, -1, 2, $&$0$&$, -2, 1, -1, \leavevmode\cancel {2}$$\rightarrow$$-3$\\
\hline
\end{tabular}
}
\end{center}
\caption{Failure of the  natural window condition \eref{e:window} for
  $p/q = 1/5$, $1/6$, $1/7$, $2/5$. There are $2q-1$ gaps, and in some
  of them the  natural window condition does not predict the correct value of $\sh $. For example, the
  first wrong prediction is $-2$, but the correct value of $\sh $ is $3$.
 }\label{tab:table1}
\end{table}

This is evidence
to the fact that there are topological
obstructions to the adiabatic deformation of one lattice into the other, {\it{i.e.}},
spectral gaps are in general closing when the lattice is deformed as a
consequence of the von Neumann-Wigner theorem. We expect such problems
to occur in other lattices as well.

A way around this problem would
be to compute $\sh $ using the Streda formula \cite{streda}.   
This, however, leads to a computational cost  
which grows exponentially with $q$ which and limits in turn   
the resolution of the output
image. Several other attempts towards this objective have been made
\cite{macdonald2, hatsugai2, avron4} but no conditions have been found yet that are both
general and simple to apply.  
As a result the computation of the colored Hofstadter butterfly still requires some effort.

\section{The Chern number through bulk-edge correspondence}

Instead of considering a bulk honeycomb lattice (infinite in both spatial dimensions) we shall consider one which is infinite in one spatial direction and semi-infinite in the other: An edge honeycomb lattice. The edge profile is assumed to be of zigzag-type, \textit{i.e.}, as the left edge in \fref{fig:graphene}.
As a consequence, instead of a 2-dimensional periodic Hamiltonian $\ham$, we consider
a Hamiltonian $\widehat{\ham}$ whose action is restricted to the half-space
Hilbert space $\widehat{\mathcal{H}} = \ell^2(\mathbb{N} \times
\mathbb{Z};\mathbb{C}^2)$. Still, $\widehat{\ham}$ can be
fibered by Bloch decomposition in the unbroken symmetry direction,
with fibers $\widehat{\ham}(k)$. Since $\ham(k) =
\widehat{\ham}(k) 
\oplus \widehat{\ham}(k) \oplus R(k)$ with $R(k)$
being a finite rank 
perturbation, $\ham(k)$ and $ \widehat{\ham}(k)$ share the same essential
spectrum:
\begin{equation*}
\label{ }
\sigma_{\rm ess}(\widehat{\ham}(k)) = \sigma_{\rm ess}(\ham(k))~.
\end{equation*}

Hence, the two spectra differ at most by a discrete spectrum in the band
gaps of $\ham(k)$. As a function of $k$, the discrete eigenvalues of 
$\widehat\ham(k)$ give rise to lines as shown in \fref{fig:edge}. 
Collectively they form the edge spectrum. 
By bulk-edge correspondence \cite{hatsugai},
the (signed) number of such eigenvalues crossing a fixed energy 
in a spectral gap of $\ham$ equals the Chern number of that gap. 
Physically, that energy is the Fermi
energy $E_{\rm F}$ and that number the Hall conductance $\sh $ (in units of $e^2/h$). 
\subsection{Detection of edge states}

The task of finding the edge spectrum at an energy $E$ in a gap of the
bulk spectrum  
can be solved based on methods developed in
\cite{ASV,porta}. 
For the convenience of the reader we adapt the proofs to the present,
simpler setting. The Hamiltonian $\widehat \ham$ is $1$-periodic in the 
$n$-direction and $q$-periodic in the $m$-direction (in
the half-space). 
In analogy to \eref{e:ham}, the Schr\"odinger equation for $\widehat\ham(k)$ can be written as
\begin{equ}
\label{e:schroed}
 \begin{pmatrix} a_m \psi_{m}^B(k)   + \psi_{m+1}^B(k) \\ \bar{a}_m\psi_{m}^A(k) + \psi_{m-1}^A(k) \end{pmatrix} = E \begin{pmatrix} \psi_{m}^A(k) \\ \psi_{m}^B(k) \end{pmatrix}~,
\end{equ}
with  $a_m =a_m(k)= e^{i  k - 2 \pi i m p/q} + 1$.
Solving with respect to $\tpsi_{m}^A$, $\tpsi_{m+1}^B$ we obtain, analogously to \cite{hatsugai2, bernevig}, 
\begin{equ}
\begin{pmatrix} 
\psi_{m+1}^B(k) \\ \psi_{m}^A(k) \end{pmatrix} = \TT_{m}^E(k) \begin{pmatrix}  
\psi_{m}^B(k) \\ \psi_{m-1}^A(k) \end{pmatrix}~,
\end{equ}
with
\begin{equ}
\label{e:31a}
\TT_{m}^E(k) = \frac{1}{\bar{a}_m}\begin{pmatrix}
       {E}^2-|a_m|^2  & - E~ \\
     E &  -1~
\end{pmatrix}~.
\end{equ}

Combining these matrices we define the transfer operator
over a period of $q$ steps (in the $m$-direction) by
\begin{equ}
\label{e:31b}
\TT^E(k) =  \prod_{m = 1}^q \TT_{m}^E(k)~=
~\TT_{q}^E(k)  \cdots \TT_{1}^E(k)~.
\end{equ}
With this definition we have
\begin{equ}
\label{e:31c}
\begin{pmatrix} {\psi}_{m+
q}^B(k) \\ {\psi}_{m + q - 1}^A(k) \end{pmatrix} = \teb(k)
\begin{pmatrix} {\psi}_{m}^B(k) \\ {\psi}_{m-1}^A (k)\end{pmatrix}~.
\end{equ}
\begin{figure}[h]
  \centering
     \def\svgwidth{0.55\textwidth}
  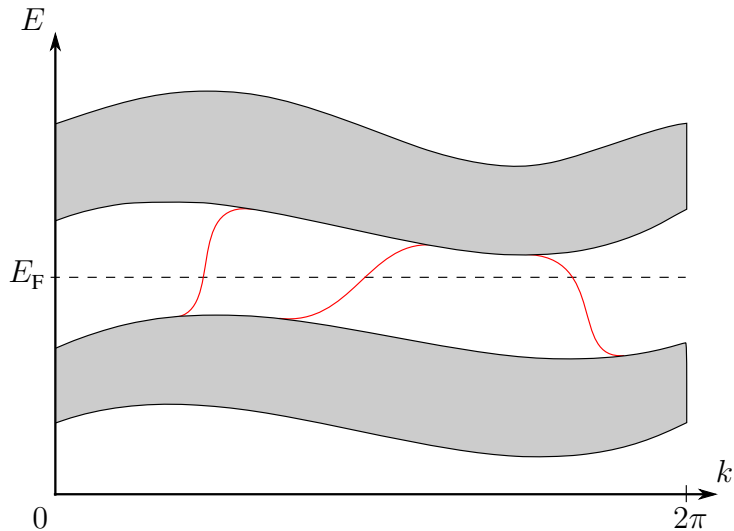
  \caption{Intersections of the edge-spectrum (red solid lines) with
    a fixed level of the Fermi energy $E_{\rm F}$ (dashed line). There are 2
    positive intersections and a negative one.}
     \label{fig:edge}
\end{figure}

Let $E$ be an energy in the gap, so that $\teb(k)$ is hyperbolic, in the sense that its eigenvalues satisfy $|\lambda| \neq 1$. An edge state 
of energy $E$ is a solution of \eref{e:schroed} and thus of
\eref{e:31c} that vanishes  
at the boundary ($\psi_{-1}^A(k) = 0$)   and decays for $m \rightarrow \infty$.
It exists \textit{iff} the contracting eigenvector $\Omega(k) = \pare{{\psi}_{0
}^B(k), {\psi}_{- 1}^A(k)}$ satisfies the boundary condition, \textit{i.e.}, {\textit{iff}}
\begin{equ}
\label{e:bc}
\Omega(k) \sim \begin{pmatrix} 1 \\0 \end{pmatrix}~,
\end{equ}
where $\sim$ stands for proportionality. 
Edge state energies occur as branches $E(k)$ depending continuously
(in fact, analytically) on $k$. 
\begin{lemma}\label{l:equal}
The count $N$ of the (signed) {\it eigenvalue} 
 crossings at a fixed energy is half the count of the (signed)
 number of times condition
 \eref{e:bc} is traversed as the {\it eigenvector} winds with $k$. 
\end{lemma}
\begin{proof}
 For the first count,
 $N$, the crossing is counted positively for decreasing
 eigenvalue branches, $E'(k) < 0$.\footnote{This seemingly unnatural
 prescription reflects the negative orientation of the edge, which has
 its outward normal pointing left.} 

To study the second count, we first note that $\TT_m=\TT_m^E(k)$ is symplectic
(for real energies $E$) in the sense that
\begin{equ}\label{e:simpl}
\TT_m^* \JJ \TT_m=\JJ~,
\end{equ}
with respect to the bilinear form given by
 $$
 \JJ = \begin{pmatrix} 0 & -1 \\ 1 & 0 \end{pmatrix} = -\JJ^*~.
 $$
Therefore, by \eref{e:31b}, we see that $\TT=\TT^E(k)$ is symplectic
as well:
\begin{equ}
 \label{e:sym}
 \TT^*\JJ\TT = \JJ~.
 \end{equ}
This is physically interpreted as current conservation. 

We next
 observe that the matrix $\teb(k)$ is real up to a factor, since the
 $\TT_m^E(k)$ in \eref{e:31a} are real up to a factor. Being $\TT$
 symplectic we have $|{\rm det}\, \TT| = 1$ and the contracting
 eigenvalue is simple. 
 Its
  eigenvector can thus be chosen to be real, $\Omega(k) = (a(k),b(k))
 \in \mathbb{R}^2$, which makes the notion of winding clear. Note that
 \eref{e:bc} will hold \emph{twice} per 
full turn of $\Omega(k)$. Therefore, we get the identity
 \begin{equ}\label{e:theformula}
 N = \int_0^{2\pi} \frac{{\rm d}k}{2 \pi i}\, \frac{\partial}{\partial k}\, \log\left(\frac{a(k) + i  b(k)}{a(k) - i  b(k)}\right)~.
 \end{equ}
  In view of the
 explanations leading to \eref{e:bc}, we need to show
 that the relative signs of the two counts are always the same. To this end
 we reinstate the dependence on $E$ in $\Omega(E,k)$ and make the
 vector locally unique by imposing $a(E,k) = 1$. Then $b(E(k),k)$ = 0,
and differentiating we get
 $$
 \frac{\partial b}{\partial E} E'(k) + \frac{\partial b}{\partial k} = 0~.
 $$
  The claim will now follow immediately from
 \begin{equ}
 \label{e:sturm}
  \frac{\partial b}{\partial E} > 0
 \end{equ}
 at crossing points, which is proven below. 
 In the following $k$ is fixed but we also suppress the dependency on
 $E$ from the notation. In addition to \eref{e:sym},
explicit calculation shows that
\begin{equ}\label{e:tmjtm}
 \TT_m^*\JJ\frac{{\rm d} \TT_m}{{\rm d} E} = \frac{1}{|a_m|^2}\begin{pmatrix}
       {E}^2+|a_m|^2  & - E~ \\
     -E &  1~
\end{pmatrix} > 0~.
\end{equ}
Using $\TT_m^* \JJ \TT_m=\JJ$ and
$$
\frac{{\rm d} \TT}{{\rm d} E} = \sum_{m = 1}^q \TT_q \cdots \frac{{\rm d} \TT_m}{{\rm d} E} \cdots \TT_1~,
$$
we conclude from \eref{e:tmjtm} that
 \begin{equ}
 \label{e:4gr}
 \TT^*\JJ\frac{{\rm d} \TT}{{\rm d} E}>0~.
 \end{equ}

From the definition of $\Omega$, we find
\begin{equ}\label{e:thisone}
\Omega^*\JJ\frac{{\rm d}\Omega}{{\rm d}E} = \bar{\psi}_{-1}^A \frac{{\rm d}\psi_0^B}{{\rm d} E} - \bar{\psi}_{0}^B \frac{{\rm d}\psi_{-1}^A}{{\rm d} E} = -\frac{{\rm d} b}{{\rm d} E}~,
\end{equ}
where the last equality holds at crossing points. 
We next 
differentiate $\pare{\TT-\lambda} \Omega = 0$, where $\lambda$ is the
contracting eigenvalue ($|\lambda| < 1$), and obtain 
$$
\pare{\frac{{\rm d}\TT}{{\rm d} E} -\frac{{\rm d} \lambda}{{\rm d} E}}\Omega + (\TT-\lambda) \frac{{\rm d}\Omega}{{\rm d} E} = 0~.
$$
Finally, we multiply the last equality from the left by
$\Omega^*\TT^*\JJ = \bar{\lambda}\Omega^*\JJ$, and using (by the reality of omega) $\Omega^* \JJ
\Omega=0$, we get
$$
\Omega^*\TT^*\JJ \frac{{\rm d}\TT}{{\rm d} E} \Omega + \pare{1 - |\lambda|^2} \Omega^*\JJ \frac{{\rm d}\Omega}{{\rm d} E} = 0~,
$$
where we used \eref{e:sym}. In view of \eref{e:4gr} and of  
\eref{e:thisone} we conclude ${\rm d} b / {\rm d} E > 0$. 
 \end{proof}
 
 As shown in the Appendix, points $k$ where $a_m = 0$ do not invalidate the above argument.

\begin{figure}[]
  \centering
     \def\svgwidth{0.68\textwidth}
  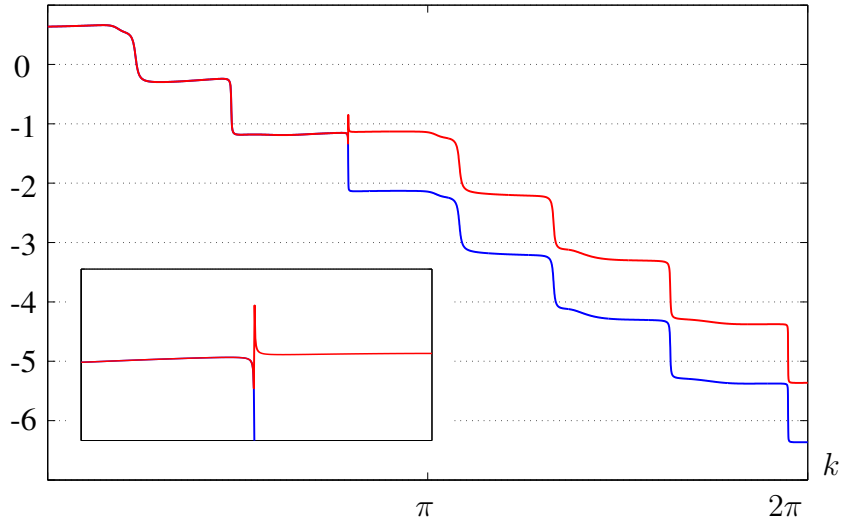
  \caption{Evolution of the phase $\theta(k)/2 \pi$ (on the $y$-axis) as a function of $k \in [0,2 \pi)$ (on the $x$-axis) for values of $p/q = 8/19$, $r = 1$. The total winding number of this function is equal to the
Chern number associated to the gap. A too low discretization of the interval (red curve, $200 \cdot q$) misses a phase turn as shown in the inset; a higher discretization (blue curve, $300 \cdot q$) resolves the problem. }
   \label{fig:winding}
\end{figure}

\subsection{Numerical implementation}

We compute the Chern number in each of the spectral gaps of the  
Hamiltonian using the r.h.s.~of \eref{e:theformula}. This involves finding the
eigenvectors of $\teb(k)$ for $E$ in a gap, and for sufficiently many
$k$ so that the variation of the phase $\theta(k)$ given by the logarithm in \eref{e:theformula} can be computed with confidence.
The issue here is that one must be able to resolve the continuity
of $\theta(k)$ when $\partial \theta / \partial k \gg 1$ as shown in \fref{fig:winding}.

\psfragscanoff
\begin{figure}[h]
  \centering
  \includegraphics[width=0.5\linewidth]{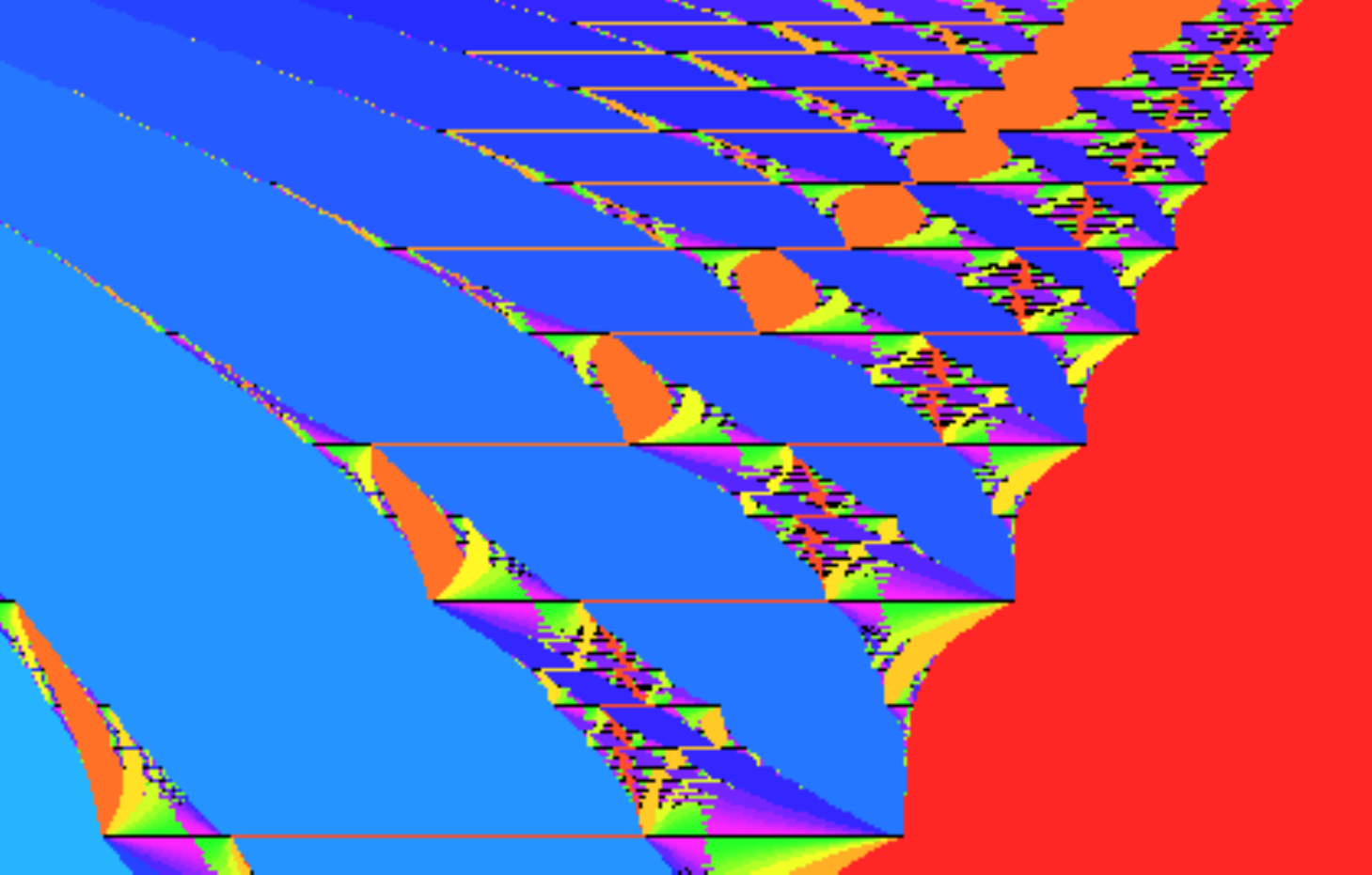}\,\includegraphics[width=0.5\linewidth]{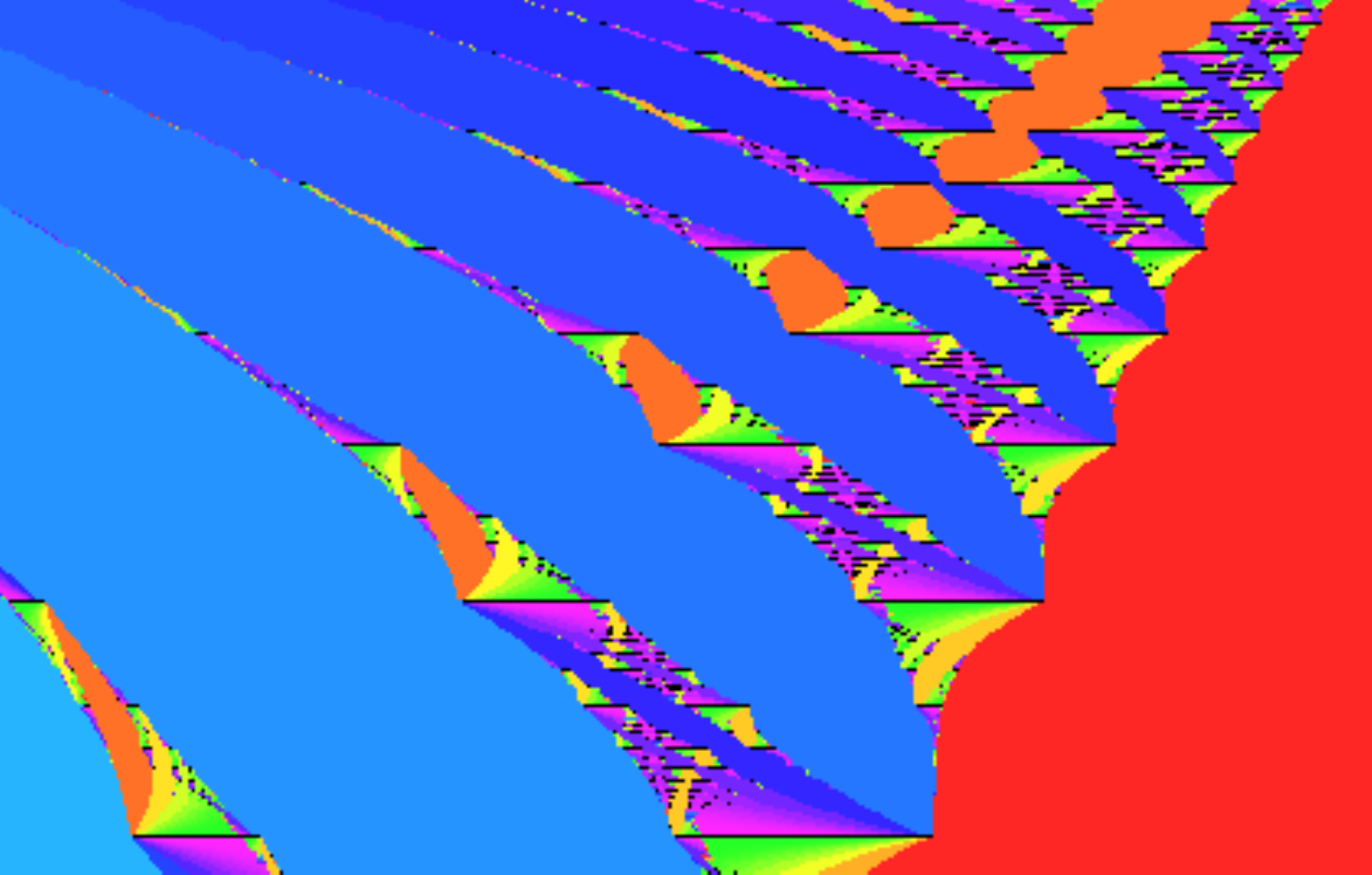}
  \caption{Comparison of the Chern numbers obtained from the  natural window
    condition \eref{e:window} (left) and the
    correct ones (right, using the bulk-edge correspondence) for the honeycomb
    lattice. Note the errors (orange horizontal lines on the left).
The horizontal axis represents the energy, the vertical axis the applied
magnetic field. Different colors in the figure represent different values of
the Chern numbers.}
   \label{fig:errors}
\end{figure}
\psfragscanon
We compute the Chern number for each gap $(e_1,e_2)$ numerically, in
the center of the gap.
We get some integer $\sh =\sh ((e_1+e_2)/2,p,q)$ that should 
be among the solutions of the Diophantine equation
\begin{equ}\label{e:dia2}
r= \sh \cdot p + s \cdot q~.
\end{equ}

If it indeed is,   we are confident that we have found the correct Chern
number. If it is not, we proceed as follows: If it is close to one of
the solutions of \eref{e:dia2}, we take that as the solution. Here,
we define close as $|\sh  -\sigma_*|/q<0.1$, where $\sigma_*$ is a
solution of \eref{e:dia2}. This leads to \fref{fig:butterfly}.

\noindent{\bf Remarks.}\\
\noindent$\bullet$ The computational cost depends on the denominator
$q$ and is of the order $\OO(q^2/\Delta k)=\OO(q^3)$, where $\Delta k$
is the discretization in $k$ which we take as $\Delta k=2\pi/(200
q)$. Going over all possible fractions $p/q$ for fixed $q$ leads to a cost of
at most $\OO(q^4)$.

\noindent$\bullet$  While the algorithm can in principle always
result in a solution of \eref{e:dia2}, there is a limitation to its
success. Whenever an eigenvalue of the matrix $\teb(k)$ comes close to
$|\lambda| = 0$ it contributes significantly to the winding number
within a small interval $\Delta k$. This may require a discretization
that is finer than reasonably doable. 

\noindent$\bullet$ The statistics of \fref{fig:butterfly} are as follows: Of the
$\OO(10^6)$ colored pixels on the figure, 
99.8\% satisfy \eref{e:dia2}, 0.1\% are merely   close to the correct Chern
number, and another 0.1\% are
undecided. The smallest gap that can be represented at this resolution is about
$e_2-e_1=0.037$ wide. We have no
rule to predict which $p$ or $q$ are most likely to lead to
difficulties.

\noindent$\bullet$ Although \eref{e:dia2} allows in principle
arbitrarily large Chern numbers, we conjecture, based on our
calculations, that \emph{the Chern numbers lie in $(-q,q)$}.

\psfragscanoff
\begin{figure}[h!]
  \centering
  \def\svgwidth{1\textwidth}
  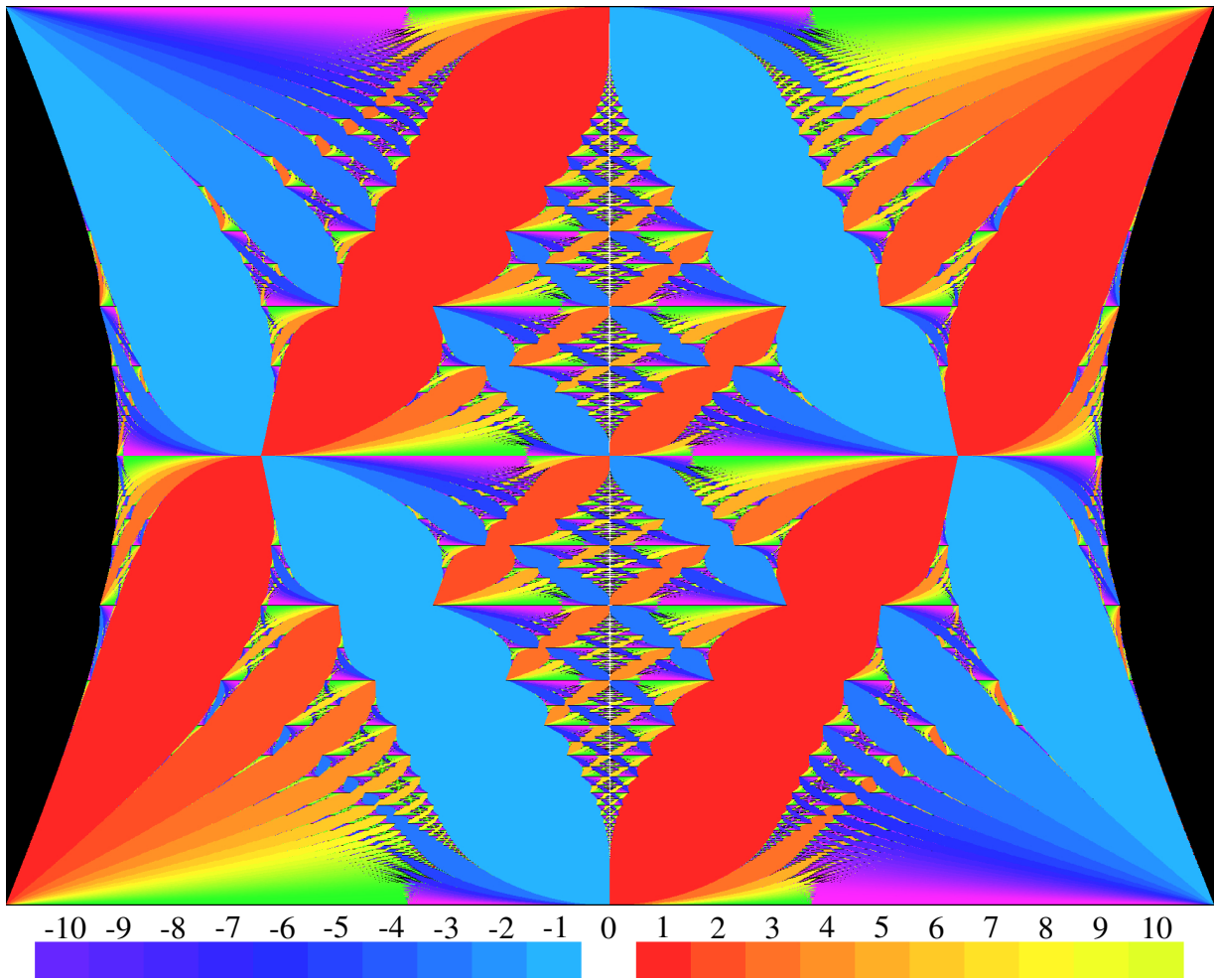
  \caption{The colored Hofstadter butterfly for the honeycomb lattice,
    as obtained  by the method of this paper. The
vertical axis is the magnetic flux per unit cell $\Phi$ ranging from $0$ to $1$. The
horizontal axis is the Fermi energy ranging from $-3$ to $3$. The colors represent
the Chern numbers. The resolution of this figure is 1920$\times$1440 and the maximal
value of $q$ is $q_{\rm max} = 720$.}
     \label{fig:butterfly}
\end{figure}
\psfragscanon

\section*{Acknowledgments}

We thank Y.~Avron for discussions and for raising the problem
addressed in this paper. The research of AA and JPE was supported in part by an ERC
Advanced Grant and that of GMG by the Swiss National Science Foundation.

\makeappendix{Appendix}

At values $k$ where $a_m = 0$ the Schr\"odinger equation cannot be
solved by transfer matrices, and we need a slight modification of the
argument.
If $a_m=0$ then\begin{equ}
\label{e:jpamtm}
\bar{a}_m \TT_m =\begin{pmatrix}
   E^2   & -E   \\
     E & -1 
\end{pmatrix} = \EE \EET~,
\end{equ}  
which has rank $1$. Instead of studying $\TT=\TT_q\cdots\TT_1$ we
study the operator $\widetilde \TT= \TT_m \widetilde \TT_m $ which is
similar to $\TT$,
with
\begin{equ}
\label{e:jpttilde}
\widetilde\TT_m= \TT_{m+q-1} \cdots \TT_{m+1}~.
\end{equ}
Because of \eref{e:jpamtm}, one eigenvalue of $\bar a_m\widetilde \TT$ is always
0. The eigenvector of the other eigenvalue of $\bar a_m \widetilde\TT$ must be
proportional to $(E,1)^{\rm T}$, and the eigenvalue is then given by
\begin{equ}
\bar a_m\TT_m\widetilde\TT_m  \EE=\EE \EET \widetilde\TT_m \EE=\lambda \EE~.
\end{equ}
We thus find a second eigenvalue $\lambda \ne0$ (and in particular no
double 0 eigenvalue, {\it i.e.}, a well-defined and continuous eigenvector $\Omega(k)$) unless
\begin{equ}
\label{e:jpperpend}
\EET \widetilde{\TT}_m \EE=0~.
\end{equ}
The next lemma shows  that this never happens, since we are
considering $E\notin \sigma(H(k))$.
\begin{lemma}
If \eref{e:jpperpend} holds then $E \in \sigma(H(k))$.
\end{lemma}
\begin{proof}
First note $\widetilde\TT_m$ is not enough to find a solution of the Schr\"odinger equation. However, in the case $a_m = 0$, \eref{e:schroed} reduces to
$$\begin{pmatrix} 
\psi_{m+1}^B(k) \\ \psi_{m}^A(k) \end{pmatrix} \sim \begin{pmatrix} E \\ 1 \end{pmatrix}~, \qquad \begin{pmatrix} 
\psi_{m+q}^B(k) \\ \psi_{m+q-1}^A(k) \end{pmatrix} \sim \begin{pmatrix} 1 \\ E \end{pmatrix}~.$$
Therefore, the condition $a_m = 0$ and \eref{e:jpperpend} imply that there is a solution of $H(k) \psi = E
\psi$ which is compactly supported on and between sites $(m,B)$, $(m+q,A)$, or
on any interval shifted by multiples of $q$. (Each such state is
normalizable.) 
Thus, $E \in \sigma(H(k))$.
\end{proof}

\bibliographystyle{JPE} 

\bibliography{Bibliography-1-1-2} 

\end{document}